\documentclass[letterpaper, 10 pt, conference]{ieeeconf}
\IEEEoverridecommandlockouts    
\usepackage{cite}
\usepackage{amsmath,amssymb,mathrsfs,amsfonts}
\usepackage{algorithmic}
\usepackage{graphicx}
\usepackage{subcaption}
\usepackage{colortbl}

\newtheorem{problem}{Problem}
\newtheorem{definition}{Definition}

\newtheorem{theorem}{Theorem}

\newtheorem{remark}{Remark}

\begin{document}
\title{\LARGE \bf Backstepping Mean-Field Density Control for Large-Scale Heterogeneous Nonlinear Stochastic Systems}

\author{Tongjia Zheng$^{1}$, Qing Han$^{2}$ and Hai Lin$^{1}$
\thanks{*This work was supported by the National Science Foundation under Grant No. IIS-1724070, CNS-1830335, IIS-2007949.}
\thanks{$^{1}$Tongia Zheng and Hai Lin are with the Department of Electrical Engineering, University of Notre Dame, Notre Dame, IN 46556, USA. {\tt\small  tzheng1@nd.edu, hlin1@nd.edu.}} 
\thanks{$^{2}$Qing Han is with the Department of Mathematics, University of Notre Dame, Notre Dame, IN 46556, USA. {\tt\small Qing.Han.7@nd.edu.}} 
}

\maketitle
\thispagestyle{empty}
\pagestyle{empty}

\begin{abstract}
This work studies the problem of controlling the mean-field density of large-scale stochastic systems, which has applications in various fields such as swarm robotics.
Recently, there is a growing amount of literature that employs mean-field partial differential equations (PDEs) to model the density evolution and uses density feedback to design control laws which, by acting on individual systems, stabilize their density towards a target profile.
In spite of its stability property and computational efficiency, the success of density feedback relies on assuming the systems to be homogeneous first-order integrators (plus white noise) and ignores higher-order dynamics, making it less applicable in practice.
In this work, we present a backstepping design algorithm that extends density control to heterogeneous and higher-order stochastic systems in strict-feedback forms.
We show that the strict-feedback form in the individual level corresponds to, in the collective level, a PDE (of densities) distributedly driven by a collection of heterogeneous stochastic systems.
The presented backstepping design then starts with a density feedback design for the PDE, followed by a sequence of stabilizing design for the remaining stochastic systems.
We present a candidate control law with stability proof and apply it to nonholonomic mobile robots.
A simulation is included to verify the effectiveness of the algorithm.
\end{abstract}


\section{Introduction}
The recent years have witnessed an enormous growth of research on the control of large-scale stochastic systems, and specific topics appear in different forms such as deployment of sensors and transportation of autonomous vehicles.
In this work, we study the problem of controlling the probability density of a large group of heterogeneous nonlinear systems.

Control problems of large-scale systems have been extensively studied by a wide range of methodologies, such as graph theoretic design \cite{olfati2007consensus} and game theoretic formulation (especially potential games \cite{marden2009cooperative} and mean-field games \cite{huang2006large}).
We pursue a strategy that directly controls the mean-field density of the systems.
This control strategy shares similar philosophies like mean-field games and mean-field type control \cite{bensoussan2013mean} in the usage of the mean-field density.
However, unlike mean-field games/control where the mean-field density is usually used to approximate the collective effect of the whole population, we aim at the direct control of this mean-field density.

Density control has been studied using discrete- and continuous-state models.
The former relies on a partition of the state space and boils down to designing transition matrices for Markov chains \cite{bandyopadhyay2017probabilistic, djeumou2020probabilistic, de2021discrete}, which usually suffers from the state explosion issue.
Continuous-state models result in a control problem of PDEs that describe the time evolution of the density function.
Early efforts on the density control of PDEs tend to adopt an optimal control formulation, which relies on expensive numerical computation of the optimality conditions and usually only generates open-loop control \cite{foderaro2014distributed, elamvazhuthi2015optimal}.
Optimal density control is also studied in \cite{ridderhof2019nonlinear, chen2021density} by establishing a link between density control and the Schr\"odinger Bridge problem.
However, except for the linear case which adopts closed-form solutions, numerically solving the associated Schr\"odinger Bridge problem also suffers from the curse of dimensionality.
Recent efforts have sought to explicitly use the mean-field density as feedback to design closed-loop and closed-form control \cite{de2018optimal, krishnan2018transport, elamvazhuthi2018bilinear, zheng2021transporting}.
Density feedback laws are able to guarantee closed-loop stability and can be efficiently computed on board.
However, the success of density feedback design relies on the assumption that the systems are homogeneous first-order integrators (with white noise).
This assumption makes the density control strategy less applicable for many systems in practice, such as wheeled mobile robots.
Stochastic systems with heterogeneous and higher-order dynamics are difficult to handle and, to the best of our knowledge, have not been studied so far.

In this work, we aim to extend density feedback design to heterogeneous and nonlinear (in particular, strict-feedback) stochastic systems.
The control objective is to design control laws to stabilize the density of states of a collection of strict-feedback stochastic systems.
The strict-feedback form is not a restrictive requirement, because many mobile vehicle robots satisfy this form and some nonlinear stochastic systems can be converted to strict-feedback forms through coordinate transformation \cite{pan2002canonical}.
We will show that the strict-feedback form in the individual level corresponds to, in the collective level, a PDE (of densities) distributedly driven by a collection of heterogeneous stochastic systems. 
Our key idea is to perform a backstepping design which starts with a density feedback design for the PDE, followed by a sequence of stabilizing design for the remaining stochastic systems.
We note that backstepping design for stochastic systems has been widely studied; see, e.g., \cite{deng1997stochastic, liu2007decentralized}.
Unlike these works where the control objective is to stabilize each system in the individual level, our control goal is to stabilize the density of these systems, meaning that each system does not necessarily exhibit equilibrium behaviors in the individual level.
This is a non-classical control problem and requires new backstepping design algorithms. 
In summary, our contribution includes: 1) presenting a backstepping design algorithm for the density control problem of large-scale heterogeneous and nonlinear stochastic systems, 2) providing specific control laws with stability analysis, and 3) applying the design algorithm to nonholonomic mobile robots as an illustration.

The rest of the paper is organized as follow. 
Section \ref{section:preliminaries} introduces some preliminaries. 
Problem formulation is given in Section \ref{section:problem formulation}. 
Section \ref{section:backstepping density control} is our main results in which we present the backstepping design algorithm and provide specific control laws.
Section \ref{section:example} provides an example using nonholonomic mobile robots.
Section \ref{section:simulation} presents an agent-based simulation to verify the effectiveness.

\section{Preliminaries}
\label{section:preliminaries}

\subsection{Notations}\label{section:notation}
For $x\in\mathbb{R}^n$, its Euclidean norm is denoted by $\|x\|$.
Let $E\subset\mathbb{R}^n$ be a measurable set. 
For $f:E\to\mathbb{R}$, its $L^2$-norm is denoted by $\|f\|_{L^2(E)}:=(\int_{E}|f(x)|^{2}dx)^{1/2}$.
We will omit $E$ in the notation when it is clear.
The gradient and Laplacian of a scalar function $f$ are denoted by $\nabla f$ and $\Delta f$, respectively
The divergence of a vector field $F$ is denoted by $\nabla\cdot F$. 


\subsection{Input-to-state stability}
Define the following classes of comparison functions:
\begin{align*}
    \mathcal{P} &:=\{\gamma:\mathbb{R}_+\to\mathbb{R}_+|\gamma\text{ is continuous, }\gamma(0)=0,\\
    &\qquad\text{ and }\gamma(r)>0\text{ for }r>0\}\\
    \mathcal{K} &:=\{\gamma\in\mathcal{P}\mid\gamma\text{ is strictly increasing}\} \\
    \mathcal{K}_\infty &:=\{\gamma\in\mathcal{K}\mid\gamma\text{ is unbounded}\}\\
    \mathcal{VK}_\infty &:=\{\alpha\in\mathcal{K}_\infty\mid\alpha\text{ is convex}\}\\
    \mathcal{CK}_{\infty} &:=\{\alpha\in\mathcal{K}_\infty\mid\alpha\text{ is concave}\}\\
    \mathcal{L} &:=\{\gamma:\mathbb{R}_+\to\mathbb{R}_+\mid\gamma\text{ is continuous and strictly}\\
    &\quad\quad\text{decreasing with }\lim_{t\to\infty}\gamma(t)=0\}\\
    \mathcal{KL} &:=\{\beta:\mathbb{R}_+\times\mathbb{R}_+\to\mathbb{R}_+\mid\beta(\cdot,t)\in\mathcal{K},\forall t\geq0,\\
    &\qquad\beta(r,\cdot)\in\mathcal{L},\forall r>0\}.
\end{align*}

We introduce the ISS concept applicable for both finite- and infinite-dimensional deterministic systems \cite{dashkovskiy2013input}.
Let $\left(X,\|\cdot\|_X\right)$ and $\left(U,\|\cdot\|_{U}\right)$ be the state and input space, endowed with norms $\|\cdot\|_X$ and $\|\cdot\|_{U}$, respectively.
Denote $U_c=PC(\mathbb{R}_+;U)$, the space of piecewise right-continuous functions from $\mathbb{R}_+$ to $U$, equipped with the sup-norm.
Consider a control system $\Sigma=(X,U_c,\phi)$ where $\phi: \mathbb{R}_+\times X\times U_c\to X$ is a transition map.
Let $x(t)=\phi(t,x_0,u)$.

\begin{definition} \label{dfn:(L)ISS}
$\Sigma$ is called \textit{input-to-state stable (ISS)}, if $\exists\beta\in\mathcal{KL},\gamma\in\mathcal{K}$, such that
\begin{equation*}\label{eq:(L)ISS}
    \|x(t)\|_X\leq\beta(\|x_0\|_X, t)+\gamma\Big(\sup_{0\leq s\leq t}\|u(s)\|_{U}\Big),
\end{equation*}
$\forall x_0\in X,\forall u\in U_c$ and $\forall t\geq0$.
\end{definition} 


\begin{definition}\label{dfn:(L)ISS-Lyapunov function}
A continuous function $V:\mathbb{R}_+\times X\to\mathbb{R}_+$ is called an \textit{ISS-Lyapunov function} for $\Sigma$, if $\exists\psi_{1},\psi_{2}\in\mathcal{K}_\infty,\chi\in\mathcal{K}$, and $W\in\mathcal{P}$, such that:
\begin{itemize}
    \item[(i)] $\psi_1(\|x\|_X)\leq V(t,x)\leq\psi_2(\|x\|_X), ~\forall x\in X$
    \item[(ii)] $\forall x\in X,\forall u\in U_c$ with $u(0)=\xi\in U$ it holds:
    \begin{equation*}
        \|x\|_X\geq\chi(\|\xi\|_U) \Rightarrow \dot{V}(t,x)\leq-W(\|x\|_X).
    \end{equation*}
\end{itemize}
\end{definition}

\begin{theorem}\label{thm:ISS-Lyapunov function}
If $\Sigma$ admits an ISS-Lyapunov function, then it is ISS.
\end{theorem}

Now we introduce ISS for finite-dimensional stochastic systems \cite{huang2009input}. 
Consider a stochastic differential equation (SDE) on $\mathbb{R}^n$:
\begin{align}\label{eq:SDE}
\begin{split}
    &dx=f(x,t,u)dt+g(x,t,u)dW_t, \quad t\in\mathbb{R}_+ \\
    &x(0)=x_0\in X\subset\mathbb{R}^n,
\end{split}    
\end{align}
where $f:X\times\mathbb{R}_+\times U\to\mathbb{R}^n,g:X\times\mathbb{R}_+\times U\to\mathbb{R}^{n\times m}$ are sufficiently smooth, and $W_t$ is an $m$-dimensional Wiener process.
For any $V(x,t)\in C^{2,1}(\mathbb{R}^n\times\mathbb{R}_+;\mathbb{R}_+)$, define the differential operator $L$:
\begin{equation*}
    LV(x,t)=V_t+V_xf+\frac{1}{2}\operatorname{Tr}(g^TV_{xx}g).
\end{equation*}

\begin{definition}
System \eqref{eq:SDE} is called \textit{$p$-th moment ISS ($p$-ISS)} if $\exists\beta\in\mathcal{KL},\gamma\in\mathcal{K}_\infty$ such that
\begin{equation*}
\operatorname{E}[\|x(t)\|^p]\leq\beta(\|x_0\|^p,t)+\gamma\Big(\sup_{0\leq s\leq t}\|u(s)\|\Big),
\end{equation*}
$\forall x_0\in X,\forall u\in U_c$ and $\forall t\geq0$.
\end{definition}

\begin{definition}
$V\in C^{2,1}(X\times\mathbb{R}_+;\mathbb{R}_+)$ is called an \textit{$p$-ISS-Lyapunov function} for system \eqref{eq:SDE}, if $\exists\psi_{1}\in\mathcal{VK}_\infty,\psi_{2}\in\mathcal{K}_\infty,\chi\in\mathcal{K}$, $W\in\mathcal{K}_\infty$, such that:
\begin{itemize}
    \item[(i)] $\psi_1(\|x\|^p)\leq V(x,t)\leq\psi_2(\|x\|^p), ~\forall x\in X$
    \item[(ii)] $LV(x,t)\leq-W(\|x\|^p)+\chi(\|u\|).$
\end{itemize}
\end{definition}

The following theorem is based on Theorem 3.1 in \cite{huang2009input}.

\begin{theorem}\label{thm:p-ISS-Lyapunov function}
If system \eqref{eq:SDE} admits an $p$-ISS-Lyapunov function, then it is $p$-ISS.
\end{theorem}

\section{Problem formulation} \label{section:problem formulation}

We study the density control problem for a family of $N$ (heterogeneous) strict-feedback stochastic systems given by:
\begin{align}\label{eq:strict-feedback system}
\begin{split}
    & dx^i = v^i(x^i,t)dt+g_1(x^i,t)dW_t^i,\quad i=1,\dots,N\\
    & dv^i = u^idt+g_2^i(x^i,v^i,t)dW_t^i,
\end{split}
\end{align}
where 
\begin{align*}
    &x^i,v^i\in\mathbb{R}^{n}\text{: states of the $i$-th system};\\
    &W_t^i\in\mathbb{R}^m\text{: standard Wiener processes independent across $i$};\\
    &u^i\text{: control input of the $i$-th system};\\
    &g_1,g_2^i\in\mathbb{R}^{n\times m}\text{: matrix-valued $C^2$ and bounded functions}.
\end{align*}

Throughout this work, the superscription $i$ is reserved to represent the $i$-th system.
Note that we only require the upper system of \eqref{eq:strict-feedback system} to be homogeneous for different $i$.
The probability density of $\{x^i\}_{i=1}^N$ is given by 
\begin{align}\label{eq:mean-field density}
    p(x,t)\approx\frac{1}{N}\sum_{i=1}^N\delta_{X_t^i},
\end{align}
with $\delta_{x}$ being the Dirac distribution.
The control objective is to design $\{u^i\}_{i=1}^N$ such that $p$ converges to a target density.

\begin{remark}
For clarity, we restrict our attention to the case when \eqref{eq:strict-feedback system} has two stages.
However, the backstepping design to be presented generalizes to systems with more stages:
\begin{align}
    & dx^i = v_1^i(x^i,t)dt+g_1(x^i,t)dW_t^i, \quad i=1,\dots,N \label{eq:top of cascade}\\
    & dv_l^i = v_{l+1}^idt+g_{l+1}^i(x^i,\bar{v}_l^i,t)dW_t^i,\quad l=1,\dots,M_i \label{eq:middle of cascade}\\
    & dv_{M_i+1}^i = u^idt+g_{M_i+1}^i(x^i,\bar{v}_{M_i+1}^i,t)dW_t^i, \label{eq:bottom of cascade}
\end{align}
where $\bar{v}_l^i=[(v_1^i)^T,\dots,(v_l^i)^T]^T$ and $M_i+1$ is the length of stages of the $i$-th system.
The objective is then to design $\{u^i\}_{i=1}^N$ to stabilize the density of $\{x^i\}_{i=1}^N$ (the states of the first stage).
Again, only \eqref{eq:top of cascade} is required to be homogeneous.
The remaining stages \eqref{eq:middle of cascade}-\eqref{eq:bottom of cascade} can be heterogeneous with different length.
This is because in our backstepping design, only the first step is identical for all systems.
The remaining steps are performed independently for different systems.
The generalization will be made clear later.
\end{remark}

Resume our discussion on \eqref{eq:strict-feedback system}.
We treat the collection of states $\{x^i\}_{i=1}^N$ as being driven by the same continuous vector field $v(x,t)$ such that $v(x^i(t),t)=v^i(x^i(t),t)$ for all $t$.
In other words, when projecting onto the trajectory of the $i$-th system, $v$ coincides with $v^i$.
In this case, the lower equation of \eqref{eq:strict-feedback system} is understood as the time differential of $v^i(x^i(t),t)$ along the trajectory $x^i(t)$.
Note that we implicitly assume that $v^i(x^i(t),t)=v^j(x^j(t),t)$ when $x^i(t)=x^j(t)$ for all $i\neq j$ and all $t$.
This is a mild assumption.
Notice that $x^i$ are all stochastic processes and the probability of $x^i=x^j$ for some $j$ is 0.
Hence, even if the above assumption is violated, it will not cause any problem to the analysis.

We confine $\{x^i\}_{i=1}^N$ within a convex bounded domain $\Omega\subset\mathbb{R}^{n}$ with a Lipschitz boundary $\partial\Omega$.
Denote $\Omega_T=\Omega\times(0,T)$, $x=[x_1,\dots,x_n]^T$, and $v=[v_1,\dots,v_{n}]^T$.
Define
\begin{align*}
    &\sigma:=[\sigma_{jk}]_{n\times n}=\frac{1}{2}g_1g_1^{T},\\
    &\nabla_\sigma p:=\big[\sum_{k=1}^{n}\partial_{x_k}(\sigma_{1k}p),\dots,\sum_{k=1}^{n}\partial_{x_k}(\sigma_{nk}p)\big]^T.
\end{align*}

In the mean-field limit as $N\to\infty$, the mean-field density $p$ satisfies a Fokker-Planck equation given by \cite{bensoussan2013mean}:
\begin{align}\label{eq:FPK equation}
\begin{split}
    &\partial_tp=-\sum_{j=1}^{n} \partial_{x_j}(v_jp)+\sum_{j,k=1}^{n} \partial_{x_j}\partial_{x_k}(\sigma_{jk}p) \quad\text{in}\quad\Omega_T\\
    &~\quad=-\nabla\cdot(vp-\nabla_\sigma p),\\
    &p=p_0 \quad\text{on}\quad\Omega\times\{0\},\\
    &\mathbf{n}\cdot(vp-\nabla_\sigma p)=0 \quad\text{on}\quad \partial\Omega\times(0,T),
\end{split}
\end{align} 
where $p_0$ is the initial density and $\mathbf{n}$ is the outward normal to $\partial\Omega$. 
The last equation is the \textit{reflecting boundary condition}. 

Our control problem is stated as follows.
\begin{problem}[Density control]
Consider systems \eqref{eq:strict-feedback system} and the mean-field density \eqref{eq:mean-field density} which satisfies \eqref{eq:FPK equation}.
Given a target density $p_*$, design $\{u^i\}_{i=1}^N$ such that $p\to p_*$.
\end{problem}

\section{Backstepping density control}
\label{section:backstepping density control}

In backstepping design, a sequence of stabilizing functions are recursively constructed \cite{krstic1995nonlinear}.
Backstepping design for stochastic systems has been widely studied \cite{deng1997stochastic, liu2007decentralized}.
The major novelty of this work is that instead of stabilizing each $x^i$ in the individual level, we aim to stabilize $p$, the density of $\{x^i\}_{i=1}^N$, in the macroscopic level.
This non-classical control objective requires new backstepping design algorithms beyond \cite{deng1997stochastic, liu2007decentralized}.

The first step of backstepping design is identical for all systems.
Given a smooth target density $p_*(x)>0$, define $\tilde{p}(x,t)=p(x,t)-p_*(x)$.
Let $v_d$ be a virtual stabilizing control law for \eqref{eq:FPK equation} (or the upper system of \eqref{eq:strict-feedback system}) and $V_1(t)=\int_\Omega\phi(\tilde{p}(x,t))dx\geq0$ be a Lyapunov certificate with $\phi:\mathbb{R}\to\mathbb{R}$ a $C^1$ function such that
\begin{align}
    \frac{dV_1}{dt}\Bigm|_{v\equiv v_d}&=\int_\Omega-W(\tilde{p}(x,t))dx\leq0
\end{align}
for some function $W:\mathbb{R}\to\mathbb{R}_+$.
The design of $v_d$ is the subject of density control and has been increasingly studied in recent years \cite{de2018optimal, krishnan2018transport, elamvazhuthi2018bilinear, zheng2021transporting}. 
One example is given by \cite{zheng2021transporting}:
\begin{align} \label{eq:density feedback law}
    &v_d = -\frac{\alpha\nabla(p-p_*)-\nabla_\sigma p}{p},
\end{align}
where $\alpha>0$ is a constant.
We can choose $V_1=\int_\Omega\frac{1}{2}\tilde{p}^2dx$ which satisfies $\frac{dV_1}{dt}|_{v\equiv v_d}=\int_\Omega-k_p\tilde{p}^2dx$ for some $k_p>0$.

\begin{remark}
Control laws like \eqref{eq:density feedback law} are called density feedback because they explicitly depend on $p$.
The density $p$ can be estimated by classical density estimation algorithms, or the recently proposed density filtering algorithms \cite{zheng2020pde, zheng2021distributedmean}.
\end{remark}

The remaining steps of backstepping design proceed independently for different systems.
Along the trajectory $x^i(t)$, define $\tilde{v}^i(x^i(t),t)=v^i(x^i(t),t)-v_d(x^i(t),t)$.
Now, define a continuous function $\tilde{v}:\Omega_T\to\mathbb{R}^n$ such that $\tilde{v}(x^i(t),t)=\tilde{v}^i(x^i(t),t)$ and $\|\tilde{v}(x,t)\|\leq\max_i\|\tilde{v}^i(x^i(t),t)\|$ for all $(x,t)\in\Omega_T$.
By definition, $\tilde{v}$ is not unique.
However, we are only interested in its values on the trajectories $\{x^i(t)\}_{i=1}^N$, which are unique.
Using $\tilde{p}=p-p_*$, we rewrite \eqref{eq:FPK equation} as
\begin{align}\label{eq:density error equation}
    \partial_t\tilde{p} = -\nabla\cdot(v_dp-\nabla_\sigma p+\tilde{v}p),
\end{align}

By It\^o's lemma, we have, along $\{x^i(t)\}_{i=1}^N$, 
\begin{align}\label{eq:velocity error equation}
\begin{split}
    d\tilde{v}^i =& (u^i-\partial_tv_d-\partial_xv_dv^i-G)dt\\
    &+ (g_2^i-\partial_xv_dg_1)dW_t^i,
\end{split}    
\end{align}
where
\begin{equation*}
    \partial_xv_d=\Big[\frac{\partial v_{d,j}}{\partial x_k}\Big]_{n\times n}, \quad
    G =
    \begin{bmatrix}
    \frac{1}{2}\operatorname{Tr}\Big(g_1^T\big[\frac{\partial^2v_{d,1}}{\partial x_j\partial x_k}\big]_{n\times n}g_1\Big)\\
    \vdots\\
    \frac{1}{2}\operatorname{Tr}\Big(g_1^T\big[\frac{\partial^2v_{d,n}}{\partial x_j\partial x_k}\big]_{n\times n}g_1\Big)
    \end{bmatrix}.
\end{equation*}

Equations \eqref{eq:density error equation} and \eqref{eq:velocity error equation} constitute a composite system where the PDE \eqref{eq:density error equation} is distributedly driven by a collection of finite-dimensional SDEs \eqref{eq:velocity error equation}.
Now our goal is to design $\{u^i\}_{i=1}^N$ such that $(\tilde{p},\tilde{v}^i)\to0$.
For this purpose, consider the following augmented Lyapunov function:
\begin{align}\label{eq:augmented Lyapunov}
\begin{split}
    V_2^i(t) & =V_1(t)+\int_\Omega\frac{1}{4}\|\tilde{v}^i(x^i(t),t)\|^4dx, \\
    & =\int_\Omega\phi(\tilde{p}(x,t))+\frac{1}{4}\|\tilde{v}^i(x^i(t),t)\|^4dx.
\end{split}
\end{align}

We discuss how to apply the differential operator $L$ to $V_2^i$.
First, since $\tilde{p}$ satisfies a deterministic equation, $LV_1$ becomes $\frac{dV_1}{dt}$.
Next, note that $\tilde{v}^i(x^i(t),t)$ is only a function of $t$.
The term $\int_\Omega\|\tilde{v}^i(x^i(t),t)\|^4dx$ simply copies the value of $\tilde{v}^i(x^i(t),t)$ to all $x\in\Omega$ for all $t$, i.e., $\int_\Omega\|\tilde{v}^i(x^i(t),t)\|^4dx=|\Omega|\|\tilde{v}^i(x^i(t),t)\|^4$, where $|\Omega|$ is the Lebesgue measure of $\Omega$.
As a result, we can apply $L$ directly inside the integral, i.e., $L\int_\Omega\|\tilde{v}^i\|^4dx=\int_\Omega L\|\tilde{v}^i\|^4dx$.
Then we have
\begin{align*}
    LV_2^i =&\int_\Omega\phi'(\tilde{p})\partial_tp + \|\tilde{v}^i\|^2(\tilde{v}^i)^T(u^i-\partial_tv_d-\partial_xv_dv^i-G) \\
    &+\frac{1}{2}\operatorname{Tr}\big((g_2^i-\partial_xv_dg_1)^T(2\tilde{v}^i (\tilde{v}^i)^T\\
    &\qquad\qquad+\|\tilde{v}^i\|^2)(g_2^i-\partial_xv_dg_1)\big)dx,
\end{align*}
where for the first term, by the divergence theorem and Young's inequality, we have
\begin{align*}
\begin{split}
    &\quad\int_\Omega\phi'(\tilde{p})\partial_tpdx \\
    &=\int_\Omega-\phi'(\tilde{p})\nabla\cdot(v_dp-\nabla_\sigma p+\tilde{v}^ip)dx \\
    &=\int_\Omega\nabla\phi'(\tilde{p})\cdot(v_dp-\nabla_\sigma p+\tilde{v}^ip)dx \\
    &\leq\int_\Omega-W(\tilde{p}) + (\tilde{v}^i)^Tp\nabla\phi'(\tilde{p})dx\\
    &\leq\int_\Omega-W(\tilde{p}) + \frac{1}{4}\Big\|\frac{(\tilde{v}^i)^Tp\nabla\phi'(\tilde{p})}{\epsilon_1(t)}\Big\|^4+\frac{3}{4}\epsilon_1(t)^{4/3}dx\\
    &\leq\int_\Omega-W(\tilde{p}) + \frac{\|\tilde{v}^i\|^4\|p\nabla\phi'(\tilde{p})\|^4}{4\epsilon_1(t)^4}+\frac{3\epsilon_1(t)^{4/3}}{4}dx,
\end{split}
\end{align*}
for any function $\epsilon_1(t)>0$, and for the last term, by Young's inequality, we have for all $t\geq0$,
\begin{align}
    &\quad\operatorname{Tr}\big((g_2^i-\partial_xv_dg_1)^T(2\tilde{v}^i (\tilde{v}^i)^T+\|\tilde{v}^i\|^2)(g_2^i-\partial_xv_dg_1)\big) \nonumber\\
    &=3\|\tilde{v}^i\|^2\operatorname{Tr}\big((g_2^i-\partial_xv_dg_1)^T(g_2^i-\partial_xv_dg_1)\big) \label{eq:young's inequality epsilon_2}\\
    &\leq \frac{3}{2}\Big(\frac{\|\tilde{v}^i\|^4}{\epsilon_2(t)^2}\operatorname{Tr}\big((g_2^i-\partial_xv_dg_1)^T(g_2^i-\partial_xv_dg_1)\big)^2+\epsilon_2(t)^2\Big),\nonumber
\end{align}
for any function $\epsilon_2(t)>0$.
Hence,
\begin{align*}
    LV_2^i &\leq\int_\Omega-W(\tilde{p}) +\frac{3}{4}(\epsilon_1^{4/3}+\epsilon_2^2)\\
    &\quad+ \|\tilde{v}^i\|^2(\tilde{v}^i)^T\Big[\frac{\|p\nabla\phi'(\tilde{p})\|^4}{4\epsilon_1^4}\tilde{v}^i+u^i-\partial_tv_d-\partial_xv_dv^i\\
    &\quad-G+\frac{3\tilde{v}^i}{4\epsilon_2^2}\operatorname{Tr}\big((g_2^i-\partial_xv_dg_1)^T(g_2^i-\partial_xv_dg_1)\big)^2\Big]dx.
\end{align*}

In the preceding design, $\epsilon_1(t)$ and $\epsilon_2(t)$ are introduced so that in the final form of $LV_2^i$, all the undesired terms have a common factor $\|\tilde{v}^i\|^2(\tilde{v}^i)^T$. 
If we design $u^i$ as: 
\begin{align}\label{eq:backstepping control law}
\begin{split}
    u^i &= -k\tilde{v}^i+\partial_tv_d+\partial_xv_dv^i+G-\frac{\|p\nabla\phi'(\tilde{p})\|^4}{4\epsilon_1^4}\tilde{v}^i\\
    &\quad-\frac{3\tilde{v}^i}{4\epsilon_2^2}\operatorname{Tr}\big((g_2^i-\partial_xv_dg_1)^T(g_2^i-\partial_xv_dg_1)\big)^2,
\end{split}
\end{align}
where $k>0$ is a selected constant, then $LV_2^i$ becomes:
\begin{align}\label{eq:LV2}
    LV_2^i &\leq\int_\Omega-W(\tilde{p})-k\|\tilde{v}^i\|^4 + \frac{3}{4}(\epsilon_1^{4/3}+\epsilon_2^2)dx. 
\end{align}
Substituting \eqref{eq:backstepping control law} into \eqref{eq:velocity error equation}, we obtain the closed-loop system:
\begin{align}
    &\partial_t\tilde{p} = -\nabla\cdot(v_dp-\nabla_\sigma p+\tilde{v}p) \label{eq:error system p},\\
    \begin{split}\label{eq:error system v}
        &d\tilde{v}^i = -\Big[\frac{3\tilde{v}^i}{4\epsilon_2^2}\operatorname{Tr}\big((g_2^i-\partial_xv_dg_1)^T(g_2^i-\partial_xv_dg_1)\big)^2\\
    &\qquad+\frac{\|p\nabla\phi'(\tilde{p})\|^4}{4\epsilon_1^4}+k\Big]\tilde{v}^i+(g_2^i-\partial_xv_dg_1)dW_t^i, \\
    &\tilde{v}(x^i(t),t)=\tilde{v}^i(x^i(t),t).
    \end{split}
\end{align}

During the backstepping design, $\epsilon_1(t)$ and $\epsilon_2(t)$ are introduced to facilitate the design of $u^i$.
However, they also introduce two extra terms in \eqref{eq:LV2}.
Our idea is to treat $\epsilon_1(t)$ and $\epsilon_2(t)$ as disturbances and establish that $(\tilde{p},\tilde{v}^i)$ are ISS with respect to $(\epsilon_1,\epsilon_2)$.
We believe that the existence of such a Lyapunov function \eqref{eq:augmented Lyapunov} already implies ISS for $(\tilde{p},\tilde{v}^i)$ in certain probabilistic sense.
Unfortunately, we are unaware of an ISS-Lyapunov function theorem like Theorem \ref{thm:p-ISS-Lyapunov function} that is applicable to conclude stability for our system \eqref{eq:error system p}-\eqref{eq:error system v}, which is a deterministic PDE distributedly driven by a family of finite-dimensional SDEs.
The development of such a theorem is left as future work.
Nevertheless, we are able to prove ISS properties of \eqref{eq:error system p}-\eqref{eq:error system v} for a special choice of $v_d$ given by \eqref{eq:density feedback law}.
This is based on the observation that the closed-loop system \eqref{eq:error system p}-\eqref{eq:error system v} resembles a cascade system.
Hence, we can use Theorem \ref{thm:ISS-Lyapunov function} to prove ISS for \eqref{eq:error system p} and use Theorem \ref{thm:p-ISS-Lyapunov function} for \eqref{eq:error system v} separately, and combine them to obtain the final result.
We should assume $\epsilon_1(t),\epsilon_2(t)\geq c>0$ for some constant $c$ because otherwise $u_i$ would be unbounded.

\begin{theorem}
Consider system \eqref{eq:error system p}-\eqref{eq:error system v}.
Let $v_d$ be given by \eqref{eq:density feedback law}.
Assume $\epsilon_1(t),\epsilon_2(t)\geq c>0$ for some constant $c$.
Then there exist constants $\beta_j,\gamma_j>0,j=1,2,3$, such that
\begin{align*}
    \operatorname{E}[\|\tilde{v}^i(t)\|^4]\leq & \|\tilde{v}^i(0)\|^4e^{-\beta_1t}+\gamma_1\epsilon_2(t)^2,~\forall i, \\
    \operatorname{E}[\|\tilde{p}(x,t)\|_{L^2}]
    \leq & \|\tilde{p}(x,0)\|_{L^2}e^{-\beta_2t}+\max_i\|\tilde{v}^i(0)\|^2e^{-\beta_3t} \\
    & +\gamma_2\epsilon_1(t)+\gamma_3\epsilon_2(t).
\end{align*}
\end{theorem}

\begin{proof}
For \eqref{eq:error system p}, consider a Lyapunov function $V_1(t)=\int_\Omega\frac{1}{2}\tilde{p}^2dx$.
By the divergence theorem, Poincar\'e's inequality, and the fact that $\int_\Omega\tilde{p}dx=0$, we have
\begin{align*}
    \frac{dV_1}{dt}&=\int_\Omega-\alpha(\nabla\tilde{p})^2+p\nabla\tilde{p}\cdot\tilde{v}dx\\
    &\leq\int_\Omega-\alpha(\nabla\tilde{p})^2+\|\nabla\tilde{p}\|\|p\tilde{v}\|dx\\
    &\leq-\alpha(1-\theta)\|\nabla\tilde{p}\|_{L^2}^2-\alpha\theta\|\nabla\tilde{p}\|_{L^2}^2+\|\nabla\tilde{p}\|_{L^2}\|p\tilde{v}\|_{L^2}\\
    &\leq-c_1^2\alpha(1-\theta)\|\tilde{p}\|_{L^2}^2-c_1\alpha\theta\|\nabla\tilde{p}\|_{L^2}\|\tilde{p}\|_{L^2}\\
    &\quad+\|\nabla\tilde{p}\|_{L^2}\|p\tilde{v}\|_{L^2},
\end{align*}
where $\theta\in(0,1)$ and $c_1>0$ is the constant from the Poincar\'e inequality.
By the maximum principle for \eqref{eq:FPK equation}, the solution $p$ is bounded above by a positive constant $c_2$ and we have $\|p\tilde{v}\|_{L^2}\leq c_2\|\tilde{v}\|_{L^2}$.
Hence, if
\begin{equation*}
    \|\tilde{p}\|_{L^2}\geq\frac{c_2}{c_1\alpha\theta}\|\tilde{v}\|_{L^2},
\end{equation*}
then we have
\begin{equation*}
    \frac{dV_1}{dt}\leq-c_1^2\alpha(1-\theta)\|\tilde{p}\|_{L^2}^2.
\end{equation*}
By Theorem \ref{thm:ISS-Lyapunov function}, there exist constants $\lambda_1,\kappa_1>0$ such that
\begin{align*}
    &\quad\|\tilde{p}(x,t)\|_{L^2}\\
    &\leq\|\tilde{p}(x,0)\|_{L^2}e^{-\lambda_1t}+\kappa_1\|\tilde{v}(x,t)\|_{L^2}\\
    &\leq\|\tilde{p}(x,0)\|_{L^2}e^{-\lambda_1t}+\kappa_1\Big(\int_\Omega\frac{\|\tilde{v}(x,t)\|^4}{2\epsilon_1(t)^2}+\frac{\epsilon_1(t)^2}{2}dx\Big)^{\frac{1}{2}}\\
    &\leq\|\tilde{p}(x,0)\|_{L^2}e^{-\lambda_1t}+\kappa_1\Big(\int_\Omega\frac{\|\tilde{v}(x,t)\|^4}{2\epsilon_1(t)^2}dx\Big)^{\frac{1}{2}} \\
    & \quad+\kappa_1(|\Omega|/2)^{\frac{1}{2}}\epsilon_1(t),
\end{align*}
where we used the inequality $\sqrt{a+b}\leq\sqrt{a}+\sqrt{b}$ for $a,b\geq0$.
Taking expectation on both sides and using Jensen's inequality for concave functions, we have
\begin{align}\label{eq:ISS p}
\begin{split}
    & \operatorname{E}[\|\tilde{p}(x,t)\|_{L^2}] \\
    \leq & \|\tilde{p}(x,0)\|_{L^2}e^{-\lambda_1t}+\kappa_1\Big(\int_\Omega\frac{\operatorname{E}[\|\tilde{v}(x,t)\|^4]}{2\epsilon_1(t)^2}dx\Big)^{\frac{1}{2}} \\
    & +\kappa_1(|\Omega|/2)^{\frac{1}{2}}\epsilon_1(t) \\
    \leq & \|\tilde{p}(x,0)\|_{L^2}e^{-\lambda_1t}+\frac{\kappa_1|\Omega|^{\frac{1}{2}}\max_i\operatorname{E}[\|\tilde{v}^i(t)\|^4]^\frac{1}{2}}{\sqrt{2}c} \\
    & +\kappa_1(|\Omega|/2)^{\frac{1}{2}}\epsilon_1(t).
\end{split}
\end{align}
Now we study \eqref{eq:error system v}.
Along the trajectory $x^i(t)$, consider a Lyapunov function $V_3^i(t)=\frac{1}{4}\|\tilde{v}^i(t)\|^4$.
Using \eqref{eq:young's inequality epsilon_2}, we have
\begin{align*}
    LV_3^i&=-\Big[\frac{3}{4\epsilon_2^2}\operatorname{Tr}\big((g_2^i-\partial_xv_dg_1)^T(g_2^i-\partial_xv_dg_1)\big)^2\\
    &\quad+\frac{\|p\nabla\tilde{p}\|^4}{4\epsilon_1^4}+k\Big]\|\tilde{v}^i\|^4\\
    &\quad+\frac{3}{2}\|\tilde{v}^i\|^2\operatorname{Tr}\big((g_2^i-\partial_xv_dg_1)^T(g_2^i-\partial_xv_dg_1)\\
    &\leq-\Big(\frac{\|p\nabla\tilde{p}\|^4}{4\epsilon_1^4}+k\Big)\|\tilde{v}^i\|^4+\frac{3}{4}\epsilon_2^2.
\end{align*}
By Theorem \ref{thm:p-ISS-Lyapunov function}, there exist constants $\lambda_2,\kappa_2>0$ such that
\begin{equation}\label{eq:ISS v}
    \operatorname{E}[\|\tilde{v}^i(t)\|^4]\leq\|\tilde{v}^i(0)\|^4e^{-\lambda_2t}+\kappa_2\epsilon_2(t)^2,~\forall i.
\end{equation}
Combining \eqref{eq:ISS p} and \eqref{eq:ISS v}, we obtain the desired result.
\end{proof}


\begin{remark}
The backstepping design can be generalized to systems in the form of \eqref{eq:top of cascade}-\eqref{eq:bottom of cascade} by augmenting \eqref{eq:augmented Lyapunov} as \begin{equation*}
    V_{M_i}=\int_\Omega\phi(\tilde{p})+\sum_{l=1}^{M_i}\frac{1}{4}\|v_l^i-v_{d,l}^i\|^4dx,
\end{equation*}
where $v_{d,l}^i,1\leq l\leq M_i$ is a sequence of stabilizing functions.
Starting from the third step ($l\geq2$), the design objective is similar with the classic backstepping design for stochastic systems studied in \cite{deng1997stochastic, liu2007decentralized}, where the presented algorithms can be used to address the remaining steps.
\end{remark}

\section{Density control of mobile robots}
\label{section:example}

As an example, we apply backstepping density control to a group of heterogeneous nonholonomic mobile robots.
The pose of a robot is given by $q^i=[x_1^i~x_2^i~\theta^i]^T$ where $x_1^i,x_2^i$ are the coordinates and $\theta^i$ is the orientation.
After adding white noise, the complete motion equations are given by \cite{fierro1997control}:
\begin{align*}
    dq^i= &S^i(q^i)v^idt+f_1(q^i,t)dW_t^i, \quad i=1,\dots,N\\
    M^i(q^i)dv^i= &(-V_m^i(q^i,\dot{q}^i)v^i-F^i(v^i)+\tau^i)dt\\
    &+f_2^i(q^i,v^i,t)dW_t^i,
\end{align*}
where $M^i\in\mathbb{R}^{2\times2}$ are symmetric positive definite inertia matrices, $V_m^i\in\mathbb{R}^{2\times2}$ are the centripetal and coriolis matrices, $F^i\in\mathbb{R}^2$ are the surface frictions, $W_t^i\in\mathbb{R}^m$, $f_1\in\mathbb{R}^{3\times m}$, $f_2^i\in\mathbb{R}^{2\times m}$ are as in \eqref{eq:strict-feedback system}, $\tau^i\in\mathbb{R}^2$ are the inputs, $d^i$ are related to geometric structures, and $S^i$ are given by
\begin{equation*}
S^i(q^i)=
\begin{bmatrix}
\cos\theta^i & -d^i\sin\theta^i \\
\sin\theta^i & d^i\cos\theta^i \\
0 & 1
\end{bmatrix}.
\end{equation*}
All the states are assumed to be available to the controller.

In the density control problem, we are only interested in the density of the positions $\{x^i:=[x_1^i~x_2^i]^T\}$.
So $\theta^i$ will be treated as known parameters.
Let $u^i$ be an auxiliary input.
By applying the nonlinear feedback \cite{fierro1997control}:
\begin{equation}\label{eq:change of control variable}
    \tau^i=M^i(q^i)u^i+V_m^i(q^i,\dot{q}^i)v^i+F^i(v^i)
\end{equation}
and removing the equation of $\theta^i$, we obtain
\begin{align}\label{eq:nonholonomic mobile robot}
\begin{split}
    &dx^i = T^i(\theta^i)v^idt+g_1(x^i,\theta^i)dW_t^i,\\
    &dv^i = u^idt+g_2^i(x^i,v^i,\theta^i)dW_t^i.
\end{split}    
\end{align}
where $T^i$ consists of the first two rows of $S^i$, which is invertible, $g_1$ consists of the first two rows of $f_1$, and $g_2^i=(M^i)^{-1}f_2^i$.
The above equations are in the same form of \eqref{eq:strict-feedback system}.
Hence, the stabilizing density feedback is given by
\begin{equation*}
    v_d=(T^i)^{-1}\Big[-\frac{\alpha(x,t)\nabla(p-p_*)-\nabla_\sigma p}{p}\Big],
\end{equation*}
where $\alpha>0$ can be used by individual robots to adjust their velocity magnitude.
The auxiliary inputs $u^i$ can be computed according to \eqref{eq:backstepping control law}, which then generate the actual input $\tau^i$ for each robot according to \eqref{eq:change of control variable}.
By following $\tau^i$, the density of the robots' positions converge towards a target density.

\begin{figure*}[t]
\setlength{\abovecaptionskip}{0.0cm}
\setlength{\belowcaptionskip}{-0.2cm}
    \centering
    \begin{subfigure}[b]{0.24\textwidth}
        \centering
        \includegraphics[width=\textwidth]{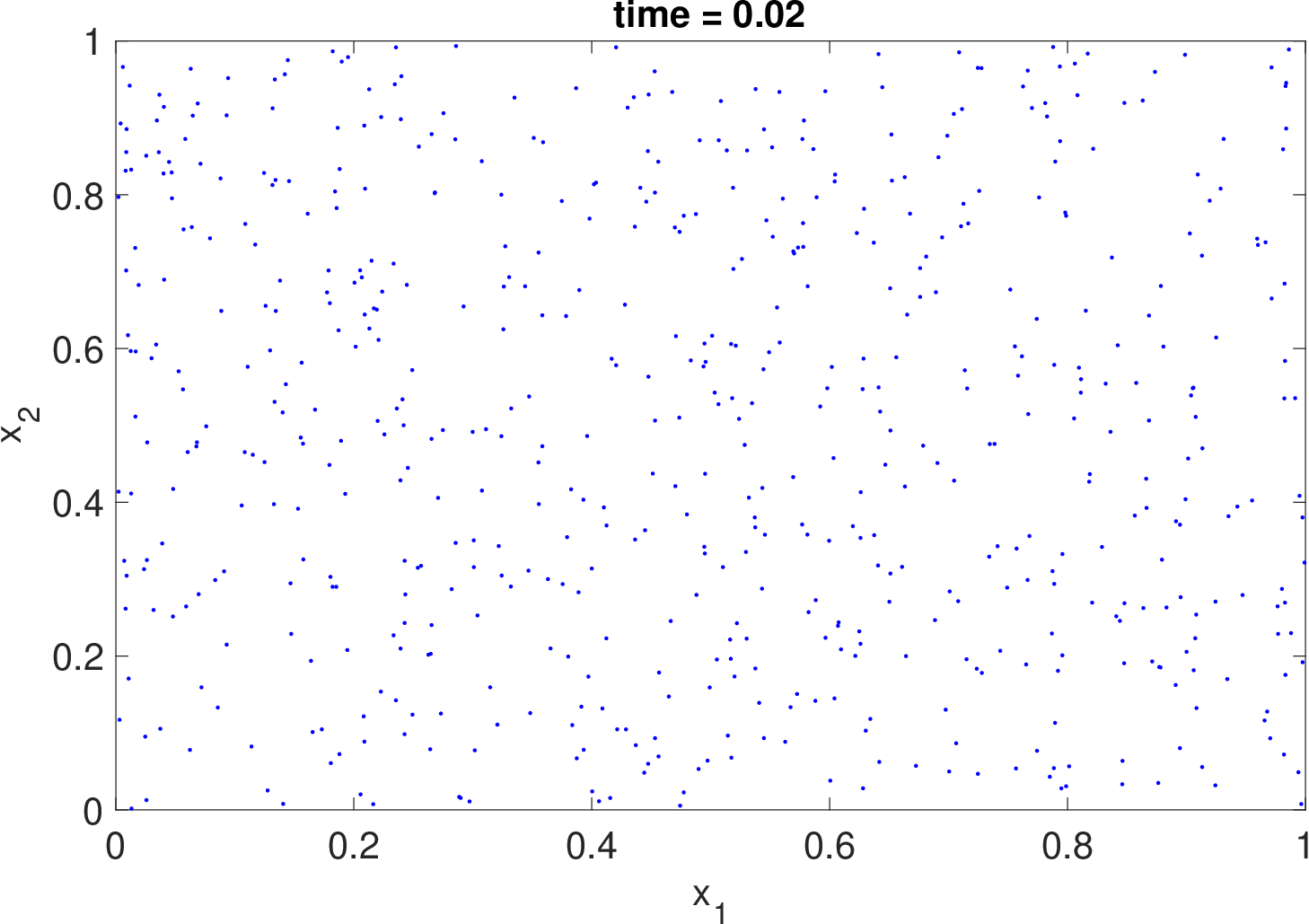}
    \end{subfigure}
    \begin{subfigure}[b]{0.24\textwidth}
        \centering
        \includegraphics[width=\textwidth]{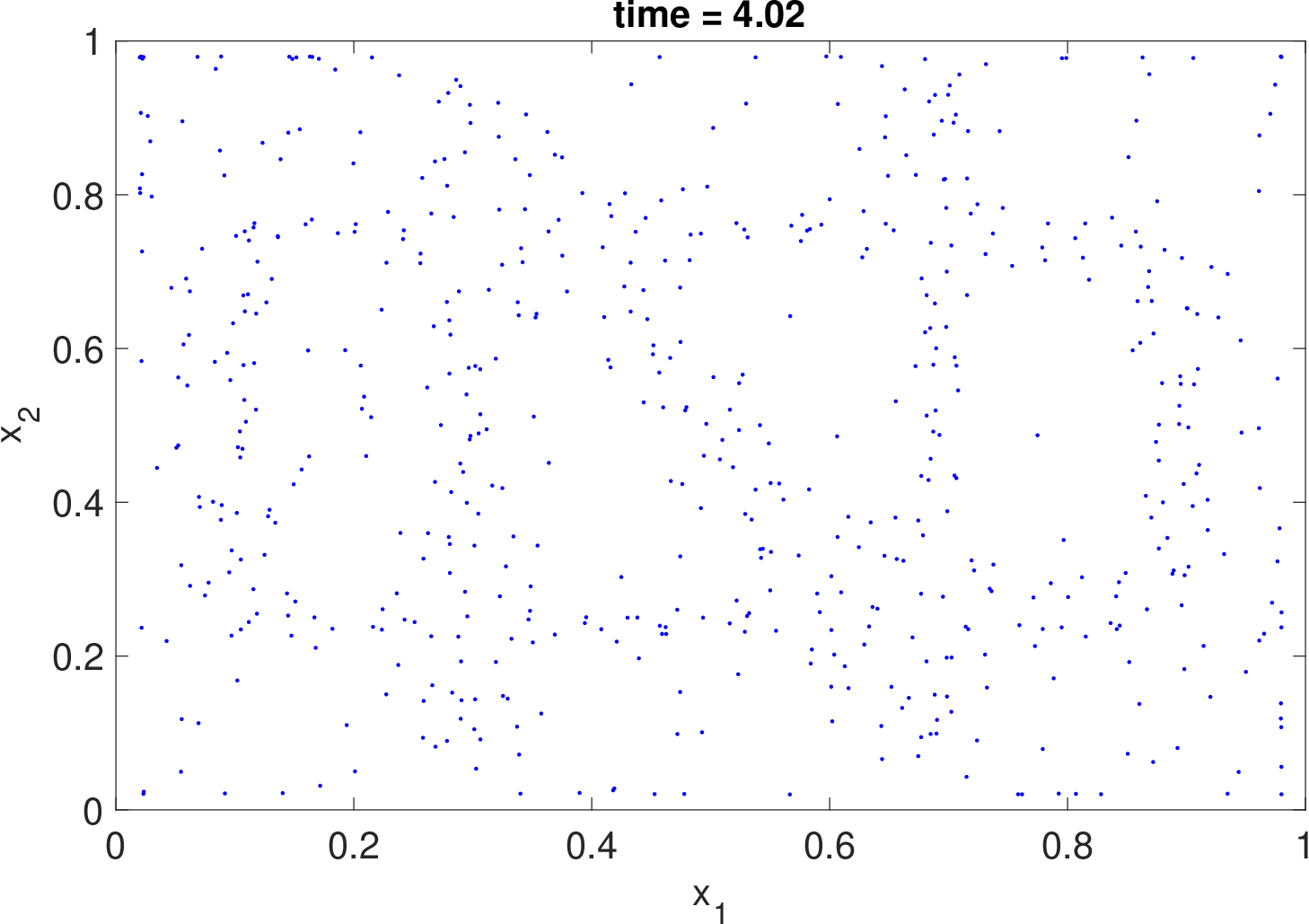}
    \end{subfigure}
    \begin{subfigure}[b]{0.24\textwidth}
        \centering
        \includegraphics[width=\textwidth]{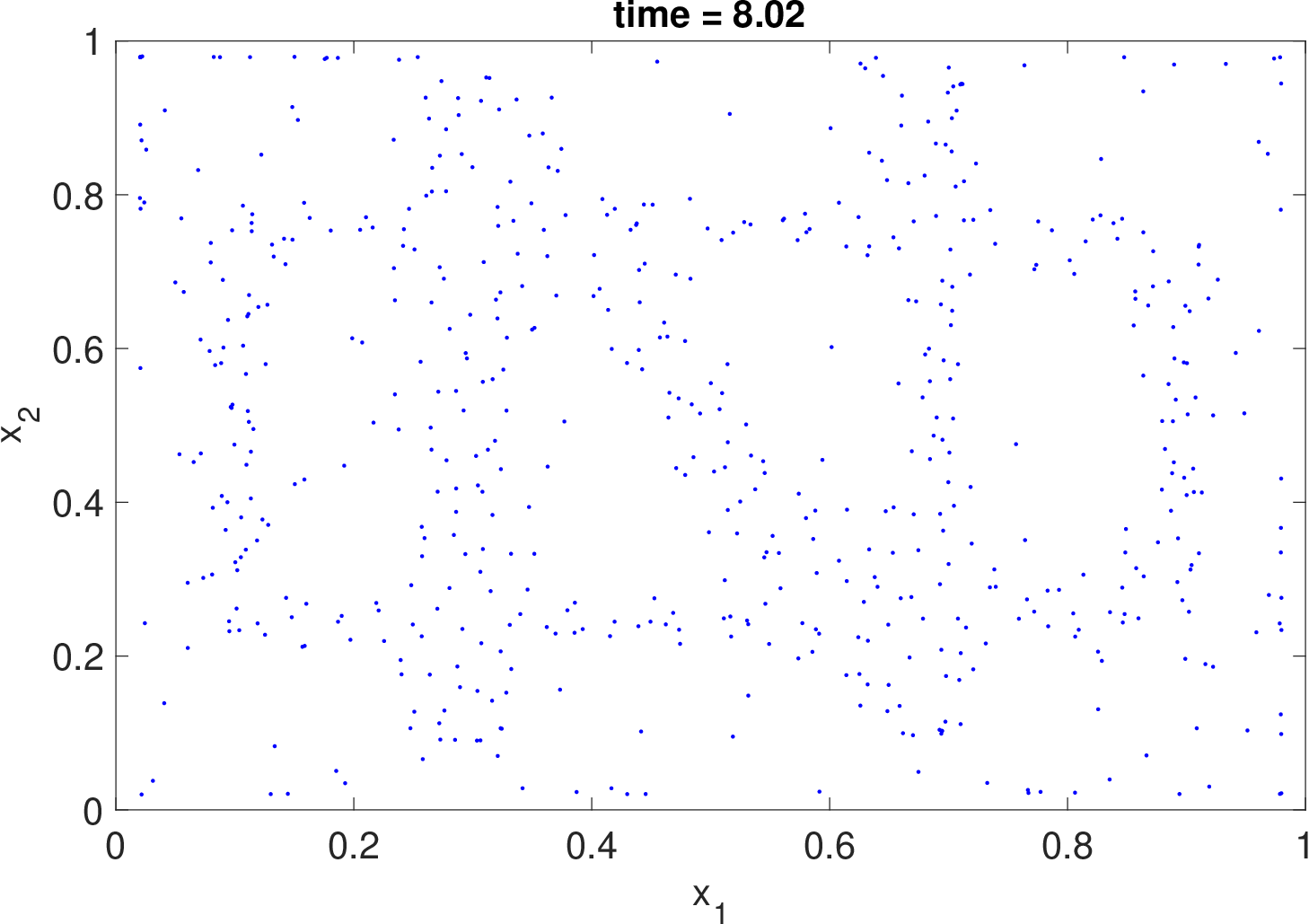}
    \end{subfigure}
    \begin{subfigure}[b]{0.24\textwidth}
        \centering
        \includegraphics[width=\textwidth]{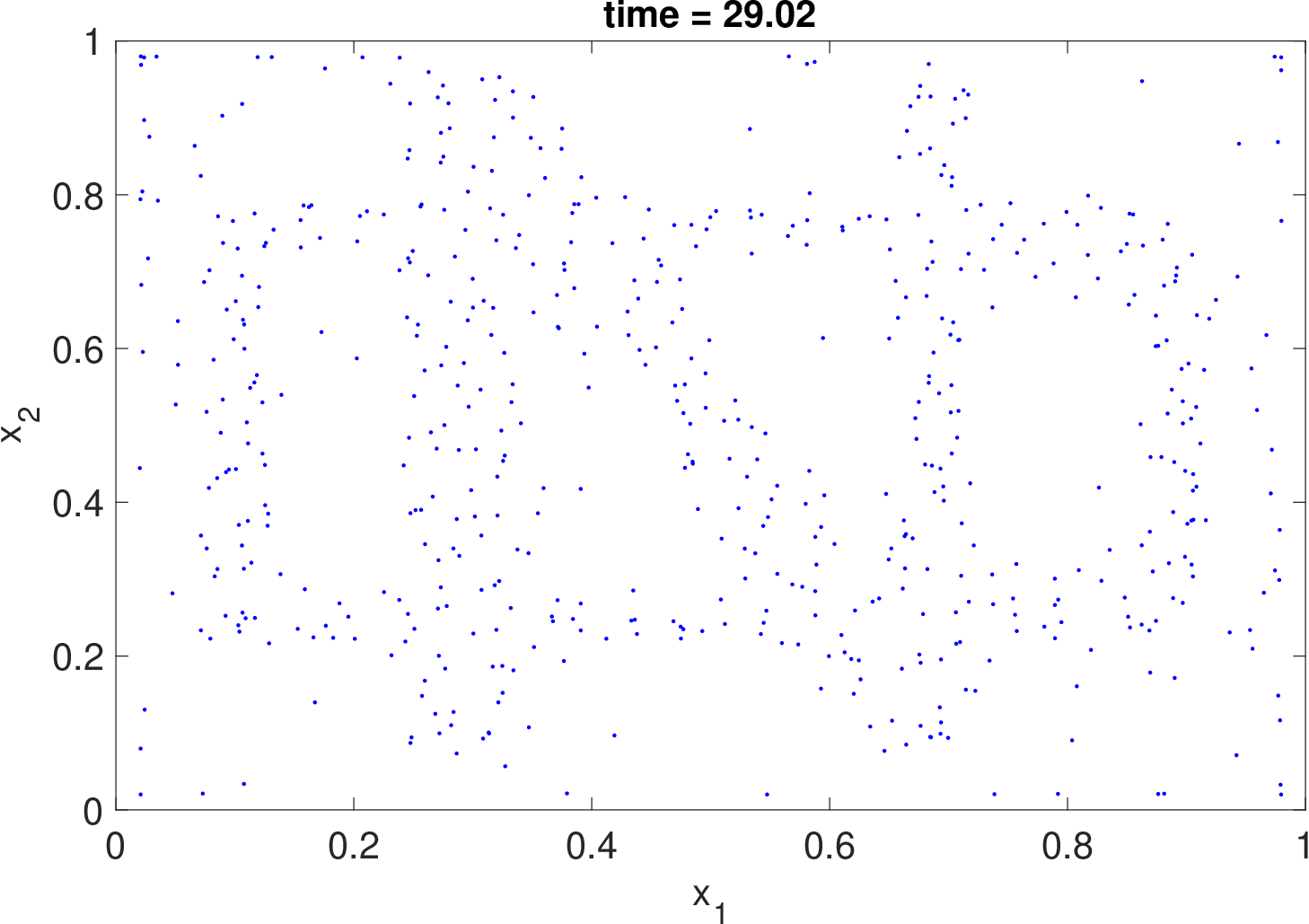}
    \end{subfigure}
    
    \begin{subfigure}[b]{0.24\textwidth}
        \centering
        \includegraphics[width=\textwidth]{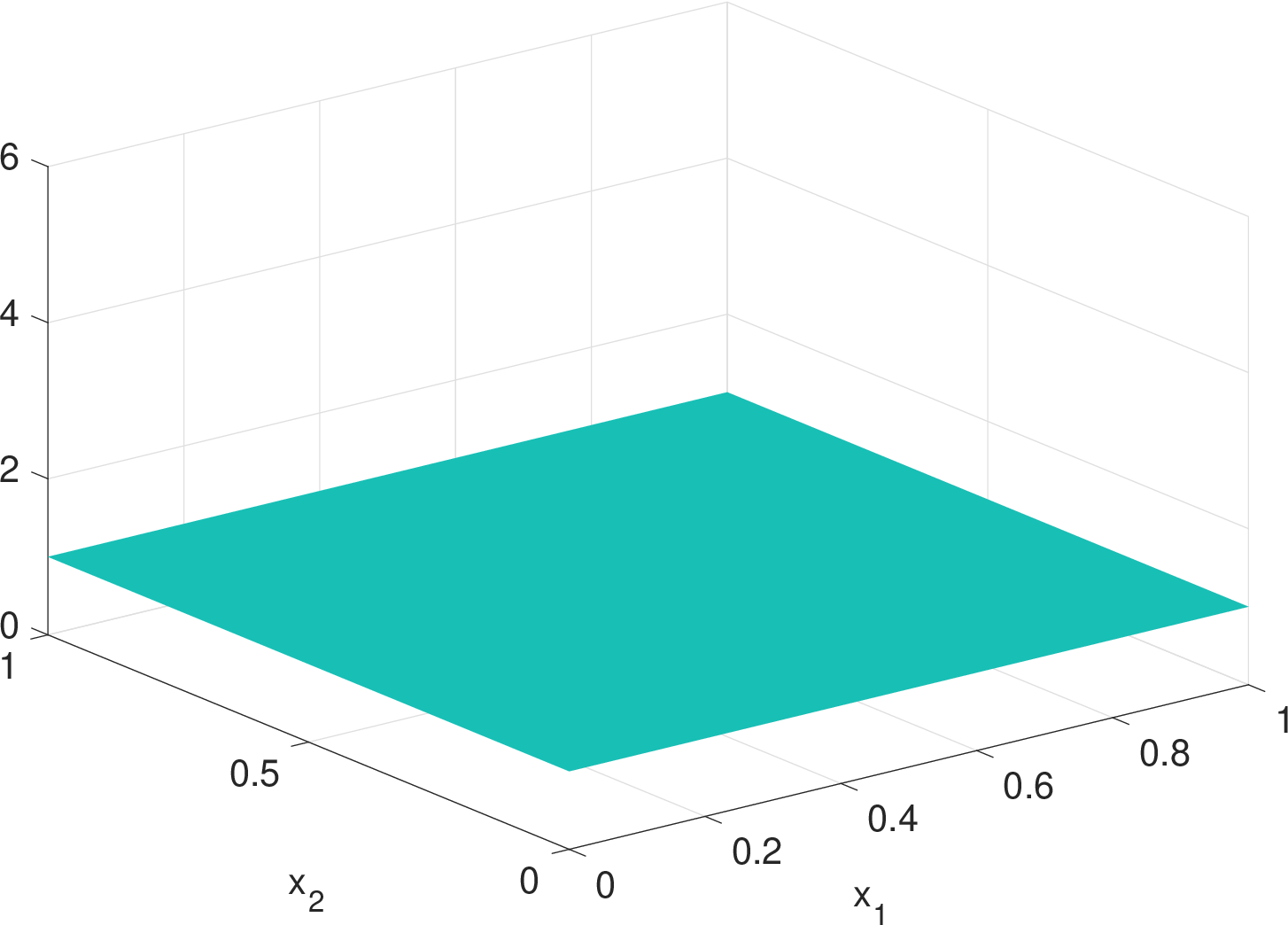}
    \end{subfigure}
    \begin{subfigure}[b]{0.24\textwidth}
        \centering
        \includegraphics[width=\textwidth]{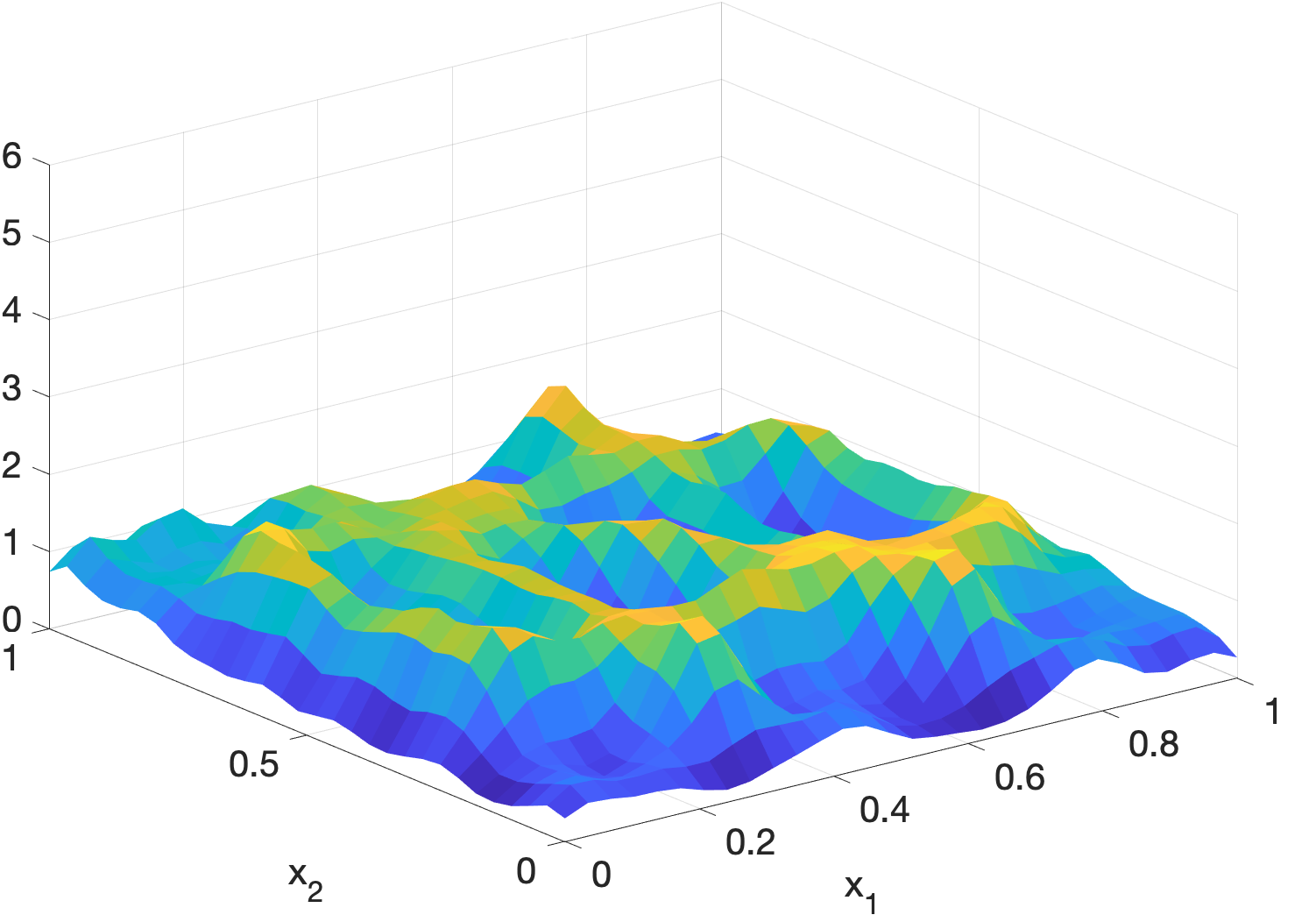}
    \end{subfigure}
    \begin{subfigure}[b]{0.24\textwidth}
        \centering
        \includegraphics[width=\textwidth]{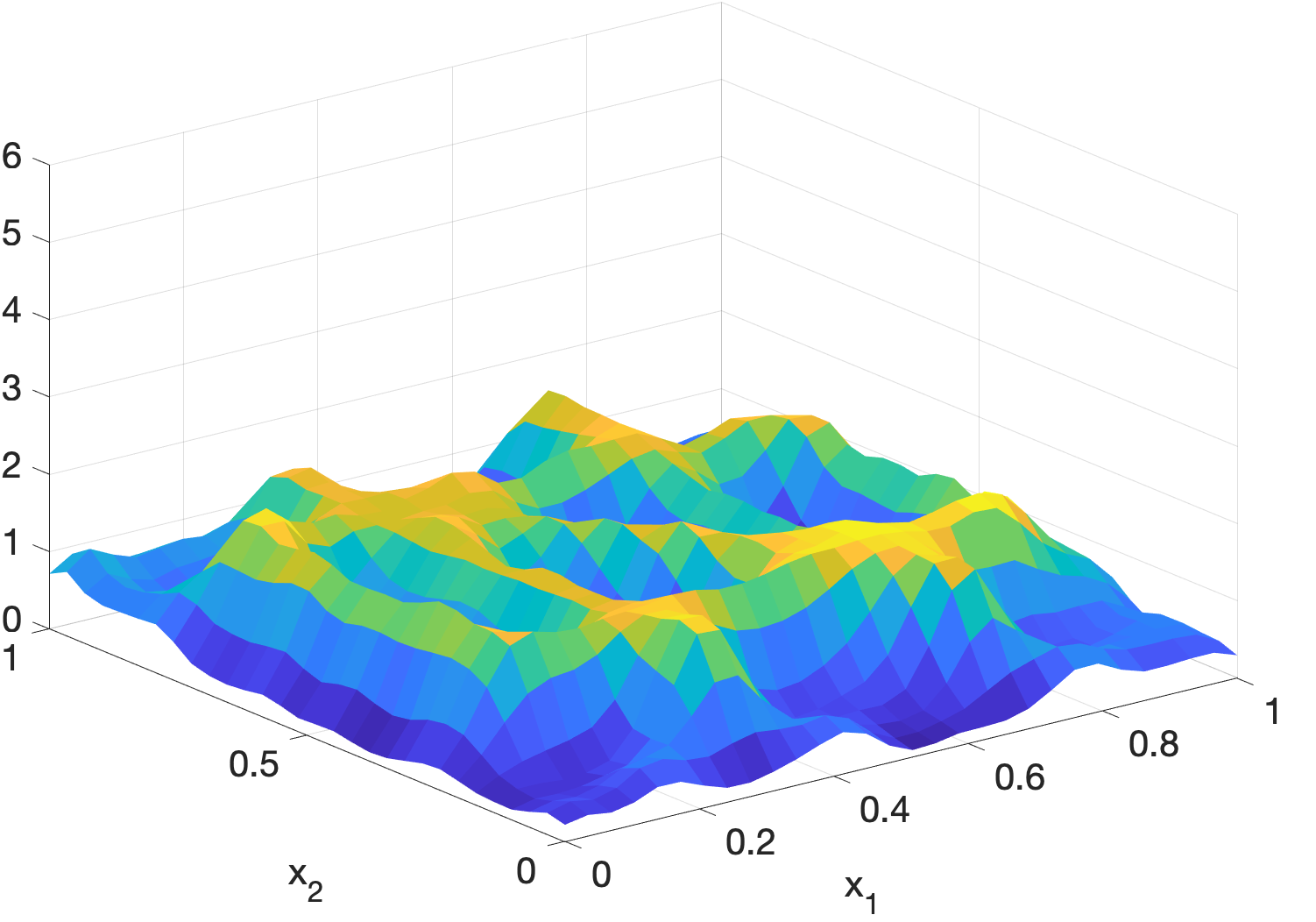}
    \end{subfigure}
    \begin{subfigure}[b]{0.24\textwidth}
        \centering
        \includegraphics[width=\textwidth]{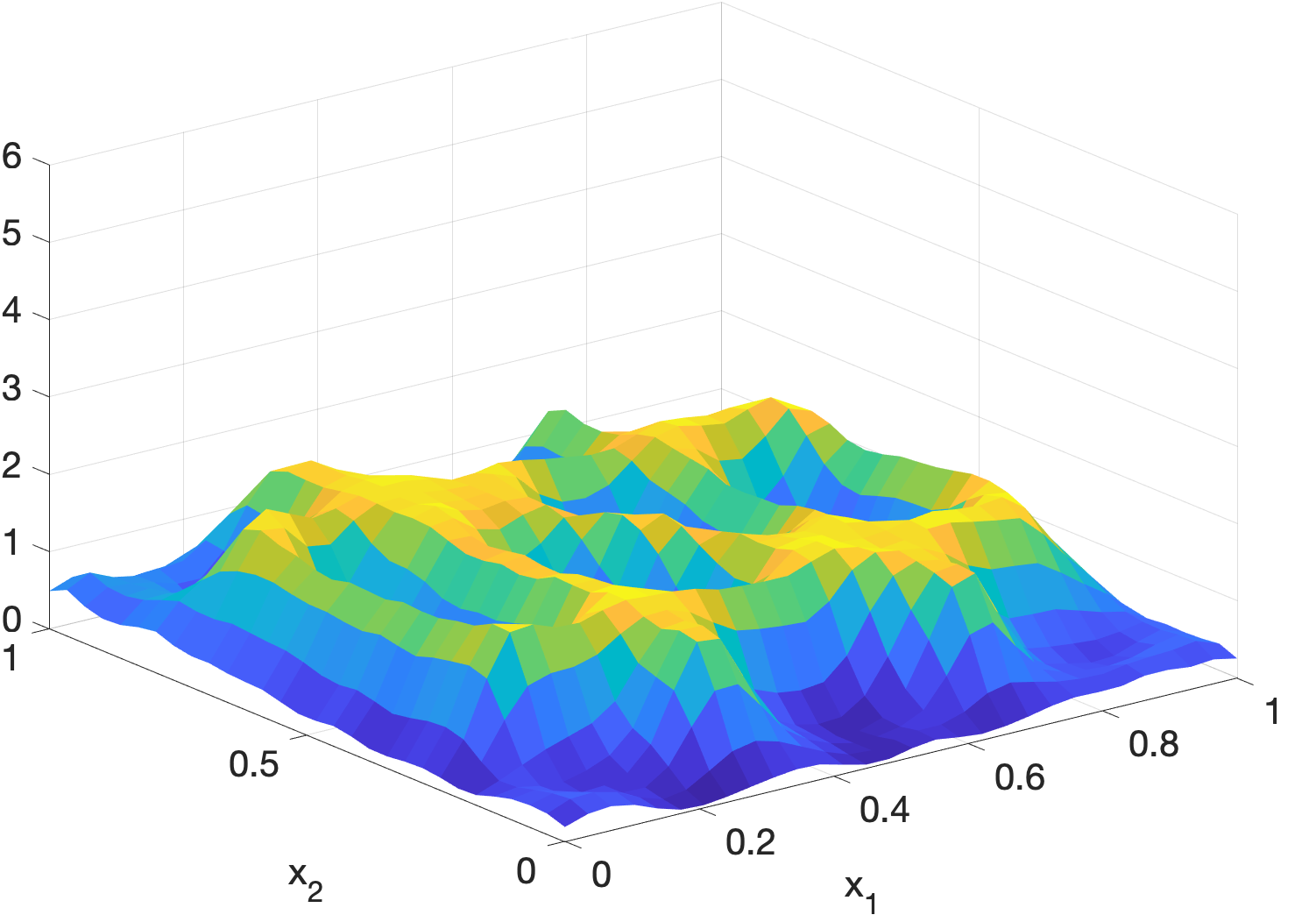}
    \end{subfigure}
    
    \caption{Top: robots' positions $\{x^i\}_{i=1}^N$. Bottom: the mean-field density $p$.}
    \label{fig:pdf evolution}
\end{figure*}

\section{Simulation study}
\label{section:simulation}
An agent-based simulation using 600 nonholonomic mobile robots is performed on Matlab to verify the proposed control law. 
We set $\Omega=(0,1)^2$.
Each robot is simulated according to \eqref{eq:nonholonomic mobile robot} where $u^i$ is given by \eqref{eq:backstepping control law} and the parameters are given by $\alpha=0.003$, $k=0.008$, $\epsilon_1(t)=\epsilon_2(t)=2$. 
Their initial positions are drawn from a uniform distribution. 
The target density $p_*(x)$ is illustrated in Fig. \ref{fig:desired density and convergence error}.
We discretize $\Omega$ into a $30\times30$ grid, and the time difference is $0.02s$. 
We use KDE (in which we set $h=0.04$) to estimate the real-time density $p$.
Simulation results are given in Fig. \ref{fig:pdf evolution}.
It is seen that the swarm is able to evolve towards the target density. 
The convergence error $\|p-p^*\|_{L^2}$ is given in Fig. \ref{fig:desired density and convergence error}, which converges and remains bounded. 

\begin{figure}[hbt!]
\setlength{\abovecaptionskip}{0.0cm}
\setlength{\belowcaptionskip}{-0.0cm}
    \centering
    \begin{subfigure}[b]{0.23\textwidth}
        \centering
        \includegraphics[width=\textwidth]{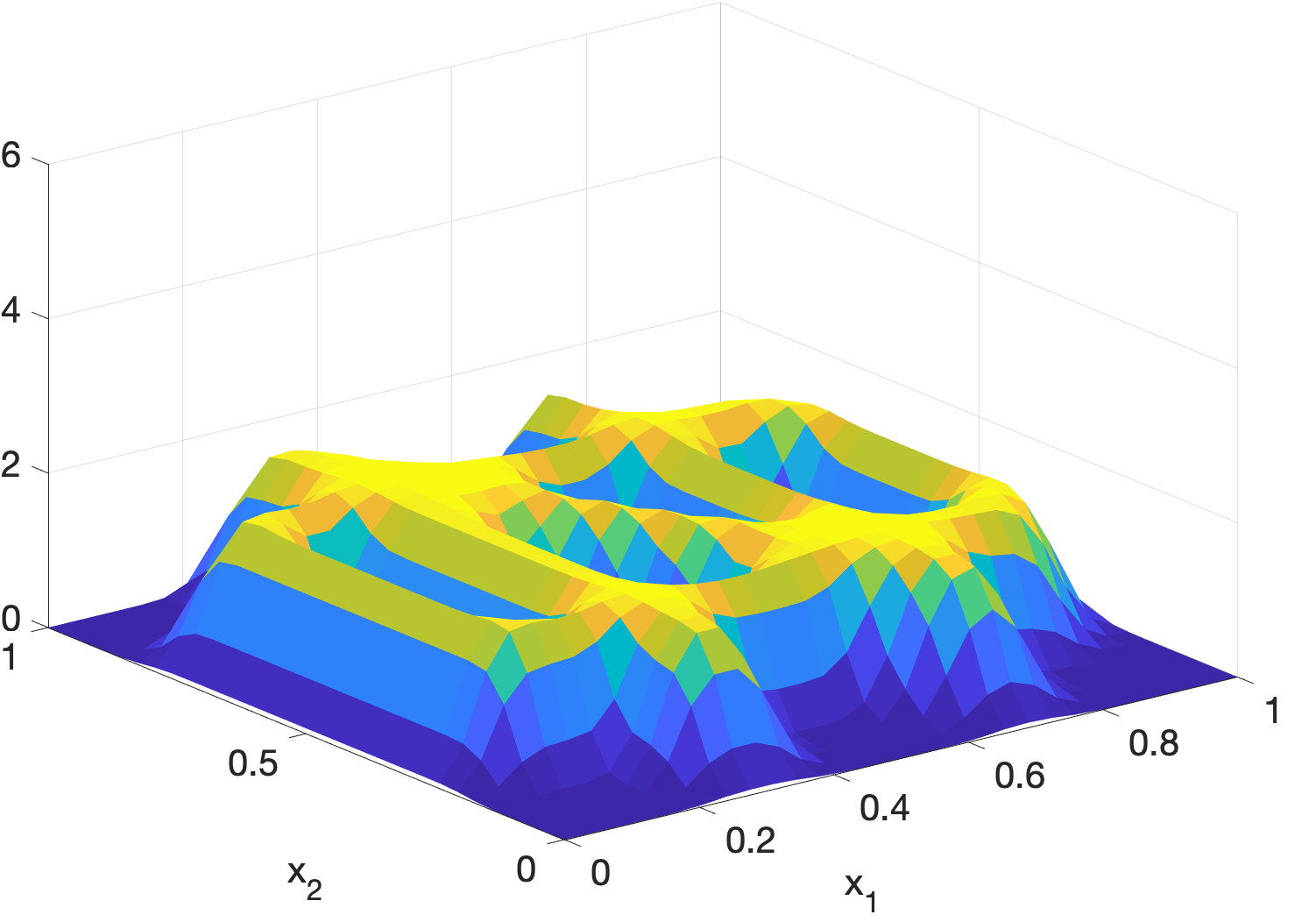}
    \end{subfigure}
    \begin{subfigure}[b]{0.23\textwidth}
        \centering
        \includegraphics[width=\textwidth]{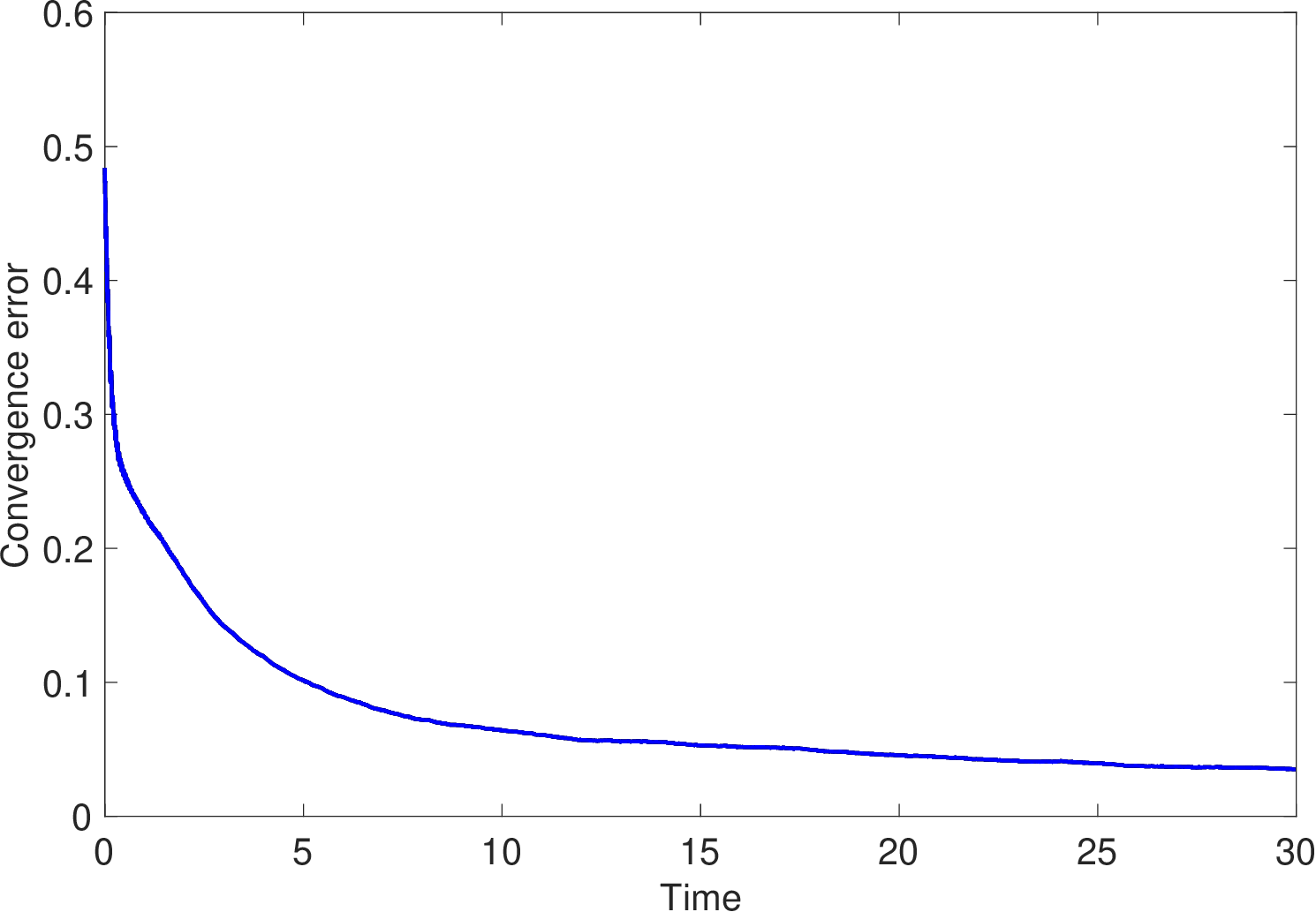}
    \end{subfigure}
    \caption{Left: desired density. Right: convergence error.}
    \label{fig:desired density and convergence error}
\end{figure}

\section{Conclusion}
This work studied the density control problem of large-scale heterogeneous strict-feedback stochastic systems.
We converted it to a control problem of a PDE that is distributedly driven by a family of heterogeneous SDEs and presented a backstepping design algorithm based on the density feedback technique. 
The presented backstepping design is suitable for many nonlinear stochastic systems including mobile robots.
We applied the algorithm to nonholonomic mobile robots and included a simulation to verify its effectiveness.
Our future work is to study the performance when the density is estimated using the density filters reported in \cite{zheng2020pde, zheng2021distributedmean}. 

\bibliographystyle{IEEEtran}
\bibliography{References}

\end{document}